\documentclass[a4paper,10pt,oneside,aps,notitlepage,nofootinbib,superscriptaddress,tightenlines]{revtex4}

\usepackage[utf8]{inputenc}
\usepackage[colorlinks,citecolor=black,linktocpage,breaklinks,hypertexnames=false,pdfpagelabels,draft=false]{hyperref}

\usepackage[intlimits]{amsmath}
\usepackage{amssymb, amsfonts,amsthm}
\usepackage{enumerate}
\usepackage[capitalise]{cleveref}
\usepackage[dvips]{graphicx}

\theoremstyle{plain}
\newtheorem{lemma}{Lemma}
\newtheorem{thm}[lemma]{Theorem}

\newtheorem{cor}[lemma]{Corollary}

\newtheorem{problem}[lemma]{Problem}

\DeclareMathOperator{\trace}{Tr}

\renewcommand{\epsilon}{\varepsilon}
\newcommand{\E}{\mathcal{E}}
\newcommand{\N}{\mathbb{N}}
\newcommand{\C}{\mathbb{C}}
\newcommand{\1}{{\openone}}

\newcommand{\be}{\begin{equation}}
\newcommand{\ee}{\end{equation}}
\newcommand{\bea}{\begin{eqnarray}}
\newcommand{\eea}{\end{eqnarray}}

\newcommand{\ra}{\rangle}
\newcommand{\la}{\langle}
\newcommand{\mc}{\mathcal}
\newcommand{\mi}{\mathrm{i}}

\usepackage{todonotes}

\newcommand{\ketbra}[3][]{\mathinner{|#2\rangle\langle #3|}_{#1}}
\newcommand{\bra}[2][]{\mathinner{\langle #2|}_{#1}}
\newcommand{\ket}[2][]{\mathinner{|#2\rangle}_{\hspace{-0.1em}#1}}
\DeclareMathOperator{\tr}{Tr}
 
\begin{document}



\title{Fundamental limitations in the purifications of tensor networks}

\author{G. De las Cuevas}
\address{Max Planck Institute for Quantum Optics, Hans-Kopfermann-Str.~1, 85748 Garching, Germany}

\author{T. S. Cubitt} 
\address{Department of Computer Science, University College London, London WC1E 6BT, UK, and\\
DAMTP, University of Cambridge, Centre for Mathematical Sciences,
Wilberforce Road, Cambridge CB3 0WA, UK}

\author{J. I. Cirac} 
\address{Max Planck Institute for Quantum Optics, Hans-Kopfermann-Str.~1, 85748 Garching, Germany}

\author{M. M. Wolf}
\address{Zentrum Mathematik, Technische Universit\"at M\"unchen, 85748 Garching, Germany}

\author{D. P\'erez-Garc\'ia}
\address{Departamento de An\'alisis Matem\'atico and IMI, Universidad Complutense de Madrid, 28040 Madrid, Spain\\
and ICMAT, C/ Nicol\'as Cabrera, Campus de Cantoblanco, 28049 Madrid, Spain}

\begin{abstract}
We show a fundamental limitation in the description of quantum many-body mixed states with tensor networks in purification form. 
Namely, we show that there exist mixed states which can be represented as a translationally invariant (TI) matrix product density operator (MPDO) valid for all system sizes, but for which there does not exist a TI purification valid for all system sizes.
The proof is based on an undecidable problem and on the uniqueness of canonical forms of matrix product states. 
The result also holds for classical states.  
\end{abstract}

\maketitle


\section{Introduction and main result}

In recent years, Tensor Network States (TNS) have become a major tool in the study of quantum many-body systems. 
Because of their ability to capture the entanglement structure present in grounds states \cite{Ha07,Sc11d,Ve08}, thermal states \cite{Mo14}, low energy states within a phase \cite{Ha13c}, or in a phase transition \cite{Vi07}, 
they constitute a powerful ansatz to describe and simulate strongly correlated quantum systems in an efficient way. 
They provide  an accurate description of the target states with a number of parameters that only grows polynomially with the size of the system, overcoming the exponential blow-up of the total Hilbert space of the system. 

Here we will restrict to a particular class of TNS for the description of mixed states in one spatial dimension (1D), the so-called Matrix Product Density Operators (MPDO) \cite{Ve04d,Zw04}. 
In the classical case (i.e.~for states diagonal in the computational basis), they  can be written as 
\be
\rho_A^L=\sum_{i_1\ldots i_L}\tr(A_{i_1}\cdots A_{i_L})\ketbra{i_1\cdots i_L}{i_1\cdots i_L}, 
\label{eq:rhoA}
\ee
where $i_j=1,\ldots , d$ and $A_i\in \mc{M}_D$, where the latter denotes the set of  $D\times D $ complex matrices. 
MPDOs play a chief role in the study of  2D systems, since they describe the so-called {\it boundary theory}, which encodes the relevant properties of the system. For this reason, they are the objects based on which one classifies the phases and phase transitions in 2D models \cite{Ci15}.
They are also relevant in the study of 1D open quantum systems  \cite{Pr08,Cu15}, in the very same way as their pure analogues ---Matrix Product States (MPS) \cite{Fa92,Sc11d,Ve08}--- are central in the study of closed 1D systems.


Ever since MPDOs were introduced, a major problem became clear: It was not easy to detect in the {\it local} matrices $A_i$ whether the global quantum state $\rho_A^L$ was positive semidefinite, as desired. This still constitutes at  one of the difficulties in the contraction of Projected Entangled Pair States (PEPS) in 2 dimensions \cite{Lu13,Lu14,We14}, and in the simulation of 1D open systems. 


In order to circumvent this problem, a particular type of MPDO was introduced \cite{Ve04d}, in which positivity was guaranteed by construction: the so-called {\it purification form}, given by a MPS with a local environment at each site, 
$$
|\Psi\rangle=\sum_{i,e} \trace(B_{i_1,e_1}B_{i_2,e_2}\cdots B_{i_L,e_L})|i_1e_1 i_2e_2\cdots i_Le_L\rangle, 
$$
where $B_{i,e}\in \mc{M}_{\tilde{D}}$. 
The resulting (unnormalized) mixed state emerges by tracing all the local environments:
\be\label{eq:purification-form}
\sigma_B^L=\trace_{e_1,\ldots e_L}|\Psi\rangle\langle \Psi|.
\ee
Note that in this case $\sigma_B^L\geq 0$ \cite{fn1} for all $L$ \emph{by construction}, whereas $\rho_A^L\geq 0$ for all $L$ \emph{by assumption}.

It is not difficult to see that, for fixed $L$ and matrices $A_i$, one can always find matrices $B_{i,e}$ of finite size $\tilde{D}$ for which $\rho_A^L/\tr(\rho_A^L)=\sigma_B^L/\tr(\sigma_B^L)$. What is not clear is whether this can be done simultaneously for all $L$. In this work we will show that, in general, this {\it cannot} be done:

\begin{thm}[Main result]
\label{thm:main}
There exist MPDOs $\{\rho_A^L\}_L$ given by matrices $(A_i)$ (meaning  that $\rho_A^L\ge 0$ for all $L$) for which there does not exist another MPDO in purification form $\{\sigma_B^L\}_L$ (see Eq.~\eqref{eq:purification-form}) such that 
\be
\label{eq:main-thm}
\frac{\rho_A^L}{\trace(\rho_A^L)}=\frac{\sigma_B^L}{\trace(\sigma_B^L)}\quad \forall  L .
\ee 
The result holds true even if we restrict ourselves to $\{A_i\}_{i=1}^7$, all $A_i$  rational-valued $7\times 7$ matrices, and the state $\rho_A^L$  a classical state (i.e.~diagonal in the computational basis) for all $L$.
\end{thm}

The proof technique is somehow non-standard since it relies on the notion of  undecidability. In \cref{sec:undec} we will show that the property 
\be
(P) \qquad \rho_A^L\ge 0 \;\; \forall L
\ee
is undecidable, in the same sense as the halting problem of a Turing machine. 
The fact that this type of problems can be undecidable was recently observed in \cite{Kl14}, where  the analogue problem for open  boundary conditions is considered. 
Note that separations between MPDOs and their purification form were given in \cite{De13c}, but in the non translationally invariant case. 
In \cref{sec:mps} we will prove some results about the canonical forms of MPS, which 
essentially show that if a purification $\{\sigma_{B}^{L}\}_{L}$  satisfying \eqref{eq:main-thm}  exists, then it can be found in a finite amount of time. 
Finally, we gather all results in \cref{sec:proof} to prove \cref{thm:main}. 
We conclude and give some outlook in \cref{sec:conclusions}.

\section{Undecidability of the positivity problem}
\label{sec:undec}

We consider the following problem related to positivity of density matrices.

\begin{problem}
\label{prob:undec}
Given $\{A_i\in \mathbb{Z}^{7\times 7}\}_{i=1}^7$, one is promised that either 
\begin{description}
\item[Case 1.] 
\label{1} $\rho_A^L\ge 0$ for all $L$; or 
\item[Case 2.] 
\label{2}
There exists $L_0$ such that $\rho_A^L\not \ge 0$ for all $L\ge L_0$.
\end{description}
Decide which case holds.
\end{problem}

\begin{lemma}
\label{lem:undec}
\cref{prob:undec} is undecidable. 
\end{lemma}

\begin{proof}
The proof of  \cref{lem:undec} relies on the recent proof that the {\it Zero in the Upper Left Corner} (ZULC) problem is undecidable for $5$  integer matrices of size $3\times 3$ \cite{Ca14}. (In \cite{Ca14} they consider rational-valued matrices, but by multiplying all of them by the least common multiple of all denominators involved, one trivially obtains the result for integer-valued matrices.)

In such ZULC problem, one is given five $3\times 3$ matrices $Y_1,\ldots Y_5$ and one is asked to decide whether there exists an $N\in \N$ and a sequence $i_1,\ldots i_N\in \{1,\ldots 5\}$ so that $$\bra{0}Y_{i_1}\cdots Y_{i_N}\ket{0}=0\,.$$
In Ref.~\cite{Ca14} it was shown that this problem is undecidable even with the additional promise that $Y_1,\ldots Y_5$ can be jointy upper-triangularized, that is, there exists a non-singular matrix $Q$ so that, all $QY_iQ^{-1}$ are upper-triangular.
In the following we reduce this problem to \cref{prob:undec}.  

Consider the projector onto the 6-dimensional symmetric subspace $P:=\frac{1}{2}(\openone + \mathbb{F})$, where $\openone$ is the identity matrix of size $3\times 3$, and  $\mathbb{F}$ the `flip' operator, 
\be
\mathbb{F}=\sum_{i,j=0,1,2}| i,j\ra\la j,i|.
\ee 
Let $\mathcal{O}$ be a real isometry of size $6\times 9$ such that 
$\mathcal{O}\mathcal{O}^\dagger = \openone$ (of size $6\times 6$)
and
$\mathcal{O}^\dagger\mathcal{O} = P$. 
We define 
\be
A_i= \mathcal{O} P (Y_i\otimes Y_i) P \mathcal{O}^\dagger \oplus \ketbra{6}{6}, \quad \textrm{for }i=1,\ldots, 5,
\label{eq:Ai}
\ee 
where $\mathcal{O} P (Y_i\otimes Y_i) P \mathcal{O}^\dagger$ acts on 
 $\C^6={\rm span}\{\ket{0},\ldots,\ket{5}\}$. 
 We also define 
\bea
&&A_6=(\ketbra{0}{0}\otimes \ketbra{0}{0})\oplus \ketbra{6}{6},\label{eq:A6}\\
&&A_7=(\ketbra{0}{0}\otimes \ketbra{0}{0})\oplus (- \ketbra{6}{6}).\label{eq:A7}
\eea
Note that $A_{i}\in \mathbb{Q}^{7\times 7}$ for all $i$.

We claim  that there is a word $j_1,\ldots, j_m$ such that 
\be
\la 0 | Y_{j_1} Y_{j_2} \ldots Y_{j_m}|0\ra =0
\label{eq:wordzulc}
\ee
if and only if $\{A_i\}$ (as defined in \eqref{eq:Ai},\eqref{eq:A6},\eqref{eq:A7}) is in Case 2 of \cref{prob:undec}. 
Our proof will also show that $\{A_i\}$ satisfies the promise of \cref{prob:undec} (i.e.~they belong either to Case 1 or 2).

For the ``only if'' direction, take the word $(j_1,\ldots, j_m)$ that satisfies \eqref{eq:wordzulc},  and define  $(i_1, \ldots, i_{L_0})=(j_1,\ldots, j_m,7)$. 
Using that $P\mathcal{O}^\dagger |0,0\ra = |0,0\ra $, that $P$ commutes with $Y_i\otimes Y_i$, and basic properties of the trace, we obtain that 
\be
\tr(A_{i_1}\ldots A_{i_{L_0}}) = \left(\la 0 |Y_{j_1}\ldots Y_{j_m}|0\ra \right)^2 - 1,
\label{eq:negtrace} 
\ee
which is negative. 
Moreover, appending  an arbitrary number of 6 at the end of $i_1, \ldots, i_{L_0}$ gives words of arbitrary length $L\geq L_0$ with a negative trace as  in \eqref{eq:negtrace}. 
Thus $\rho_A^L\ngeq 0 $ for all $L\geq L_0$.

Conversely, assume that  there is a word $i_1,\ldots, i_L$ (for some $L\geq L_0$) such that $\tr(A_{i_1}\ldots A_{i_L})<0$. 
Assume first that this word contains no 6 or 7. 
Using the above-mentioned properties of $\mathcal{O}$ and $P$, and that $\tr(\mathbb{F}(G\otimes G)) = \tr(G^2)$ for any  matrix $G$, we obtain 
\begin{equation}
\tr(A_{i_1}\ldots A_{i_L}) = 
\frac{1}{2}\left[
\left(\tr\left( Y_{i_1}\ldots Y_{i_L} \right)\right)^2 + \tr\left(\left(Y_{i_1} \ldots Y_{i_L}\right)^2\right)\right] + 1 .
\label{eq:no67}
\end{equation}
The first term in square brackets is clearly non-negative. The second is non-negative because, as mentioned above, we can assume without loss of generality that the matrices $Y_i$ can be simultaneously put in upper triangular form. Thus the word must contain at least a 6 or a 7. 

If it only contains one 6, imagine it is in the last position of the word (which can always be achieved by using cyclicity of the trace). We obtain
\be
\tr(A_{i_1}\ldots A_{i_{L}}) =\left(\la 0| Y_{i_1}\ldots Y_{i_{L-1}}|0\ra \right)^2 + 1,
\ee
which is always positive. An analogous argument holds if there are multiple 6s. 

Hence the word must contain at least one 7, and imagine again that there is just one 7 and this is the last element of the word. 
In this case, 
\be
\tr(A_{i_1}\ldots A_{i_{L}}) = \left(\la 0 |Y_{i_1}\ldots Y_{i_{L-1}}|0\ra \right)^2 - 1,
\ee
which is negative by assumption. Since the $Y_{i}$'s contain only integer numbers, this implies that $\la 0 |Y_{i_1}\ldots Y_{i_{L-1}}|0\ra =0$, as we wanted to show. An analogous argument holds if there are multiple 7s.

As commented above, by multiplying all matrices $A_i$ by the least common multiple of all denominators involved in them, one can assume without loss of generality that they have integer coefficients instead of rational ones.
\end{proof}

\section{Canonical forms of Matrix Product States}
\label{sec:mps}

In this Section we will show a result about canonical forms of MPS. This will allows us to argue in \cref{sec:proof} that determining whether  there exists a purification that satisfies \eqref{eq:main-thm} is decidable. The following results are based on Ref.~\cite{Ci15}, but here we will give the explicit bounds. 
Since we will only consider MPDOs diagonal in the computational basis, we will consider them as MPS by mapping $|i\ra\la i|$ to $|i\ra$. That is, the pure state corresponding to
$$
\rho_A^L=\sum_{i_1\ldots i_L}\tr(A_{i_1}\cdots A_{i_L})|i_1\cdots i_L\ra \la i_1\cdots i_L|
$$ 
is
\be
\label{eq:psiAL}
|\psi_A^L\ra = \sum_{i_1\ldots i_N}\tr(A_{i_1}\cdots A_{i_L})|i_1\cdots i_L\ra .
\ee
We also denote by $|\psi_C^L\ra$  the pure state corresponding to the MPDO in purification form associated to $\{B_{i,e}\in \mathcal{M}_{\tilde{D}}\}$ (see Fig.~\ref{fig:states-fundamental-limitations}). More precisely,
\be
|\psi_C^L\ra = \sum_{i_1\ldots i_N}\tr(C_{i_1}\cdots C_{i_L})|i_1\cdots i_L\ra 
\label{eq:psiCL}
\ee
with  
\be
\label{eq:pure-purification}
(C_{i})_{(\beta,\gamma),(\beta',\gamma')} = 
\sum_{k} (B_{i,k})_{\beta,\beta'}  (\bar{B}_{i,k})_{\gamma,\gamma'}.
\ee
Note that the size of the matrices $C_i$ is $D'=\tilde{D}^2$.
 
 \begin{figure}[t]
\includegraphics[width=0.7\textwidth]{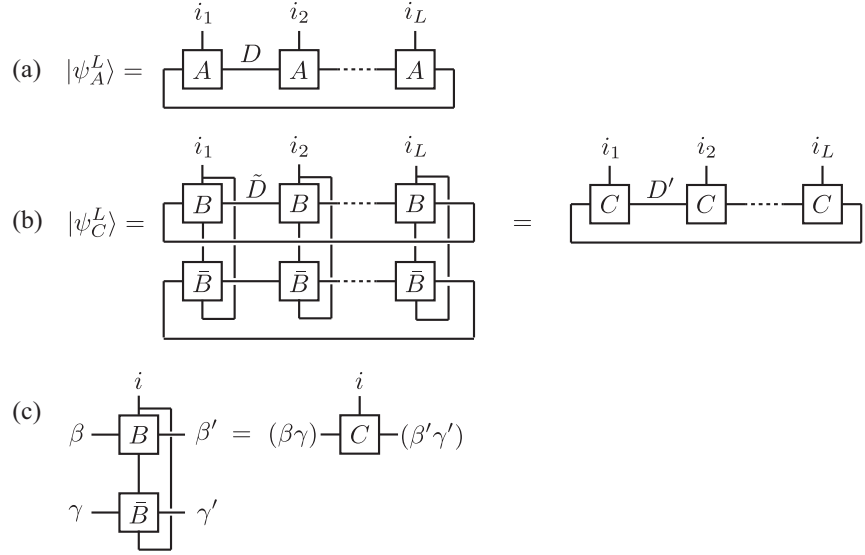}
\caption{The tensor network diagrams for  (a) $|\psi_{A}^{L} \ra$ (see Eq.~\eqref{eq:psiAL}),  
(b)  $|\psi_{C}^{L}\ra$ (see Eq.~\eqref{eq:psiCL}) and  (c) $C$ (see Eq.~\eqref{eq:pure-purification}). 
}  
\label{fig:states-fundamental-limitations} 
\end{figure}

The main result that we need concerning MPS is the following. We will state the results for general matrices $\{A_{i}\}$, $\{C_{i}\}$  but finally we will apply them to $C_{i}$ of the form \eqref{eq:pure-purification}.

\begin{thm}
\label{thm:mps}
Let  $\{A_i\in \mathcal{M}_{D}\}_{i=1}^d$, $\{C_i\in \mathcal{M}_{D'}\}_{i=1}^d$
be two families of matrices for which there exist  constants $m_L$ such that
\be
\ket{\psi_A^L} =m_L \ket{\psi_C^L} \quad \forall L\le D!D'! 3 (D+D')^6 .
\label{eq:hyp}
\ee
Then there exist $m_L$ such that 
$\ket{\psi_A^L} = m_L\ket{\psi_C^L}$ for all $L$ which are multiples of $D!D'!3(D+D')^5$.  
\end{thm}

To prove \cref{thm:mps}, we need to gather some facts about MPS. We start with the canonical forms of MPS.

\begin{thm}[Canonical form (Theorem 4 in \cite{Pe07})]
\label{thm:cf}
Given $\{\tilde{A}_i\in \mathcal{M}_{\tilde{D}}\}_{i=1}^d$, there exists another set of matrices $\{A_i\in \mathcal{M}_{D}\}_{i=1}^d$ with $D\leq \tilde{D}$ so that $\ket{\psi_A^L}=\ket{\psi_{\tilde{A}}^L}$ for all $L$ and the $A_i$ are block diagonal,  $A_i=\oplus_j \lambda_j A^j_i$,  and satisfy that for all $j$
\begin{enumerate}[(i)]
\item $\lambda_j>0$, 
\item \label{i} $\1$ is the only fixed point of the operator $\E(X)= \sum_i A^j_i XA^{j\dagger}_i$,  and 
\item \label{ii} the unique fixed point of its adjoint map $\E^*(Y)=\sum_i A^{j\dagger}_i YA^{j}_i$ is a full rank positive diagonal matrix $\Lambda_j$ (normalized so that  $\tr(\Lambda_j)=1$).
\end{enumerate}
\end{thm}

The blocks in the decomposition $A_i=\oplus_j \lambda_j A^j_i$ may not be injective yet. To make them injective, we only need to group $D!$ sites together, 
since the periodicity of each block is a  divisor of the size of the block (see Lemma 5 of \cite{Ca13b}, \cite{Fa92}). 

\begin{thm}[Quantum Wielandt \cite{Sa10}]\label{wielandt}
Consider the matrices $\{A_i\in \mathcal{M}_{D}\}_{i=1}^d$, so that the corresponding MPS becomes injective after blocking $L_0$ sites. Then $L_0$ can be taken to be $D^4-1$.
\end{thm}

Now we show that if two injective MPS are proportional for some lengths, then their matrices must be related by a unitary.

\begin{thm}[Injective case]
Assume that $\{A_i\in \mathcal{M}_{D}\}$ and $\{C_i\in \mathcal{M}_{D'}\}$ 
are such that $|\psi_{A}^{L}\ra$ and $|\psi_{C}^{L}\ra$ become injective at some length $L_0=D^4-1$, 
and both verify \eqref{i} and \eqref{ii} of \cref{thm:cf}. 
Assume that $D'\le D$. Fix $L$ so that $L\ge 2L_0+1$  and assume that $\ket{\psi_A^L}$ is proportional to $\ket{\psi_{C}^L}$. Then $D=D'$ and there exists a unitary $U$ and a phase $e^{\mi\phi}$ so that $A_i=e^{\mi\phi} UC_i U^{\dagger}$ for all $i$.
\end{thm}
\begin{proof}
Let $h$ denote the interaction hamiltonian of range $L_0+1$ that has $\ket{\psi_A^L}$ as its unique ground state (up to constant) in any chain of length $N\ge 2L_0+1$ (Theorem 10 of \cite{Pe07} and \cref{wielandt} for the injectivity bound of $\{A_i\}$). Using injectivity of $\{C_i\}$ and the hypothesis it is easy to see that $\ket{\psi_{C}^L}$ is also a ground state for $h$  for all chains of length $L\ge 2L_0+1$. Therefore $\ket{\psi_{C}^L}$ and $\ket{\psi_{A}^L}$ are proportional for all $L\ge 2L_0+1$.
Finally we use Lemma 3 of \cite{Ca13b} to conclude the proof.
\end{proof}

Now we state that each of the injective blocks is linearly independent. 

\begin{thm}[Linear independence of injective blocks (Proposition 4 of \cite{Pe07})]
\label{thm:li}
Let $\{A^1_i\},\ldots, \{A^b_i\}$ be of sizes  $D_1\ge \cdots \ge D_b$, 
so that all of them define injective MPS, 
each of them is in canonical form, 
and assume that neither of them is related by $e^{\mi\phi} U \cdot U^\dagger$. Then for all $L\ge 3(b-1)D_1^4$ the vectors $\ket{\psi_{A^1}^L},\ldots, \ket{\psi_{A^b}^L}$ are linearly independent.
\end{thm}

Finally we need to collect some facts about moments of sets of numbers. 

\begin{lemma}\label{lem1}
Let $\alpha_1,\ldots,\alpha_r,\beta_1,\ldots,\beta_n\in \C$ with $n\ge r$ so that 
\be
\sum_{i=1}^r\alpha_i^L=\sum_{j=1}^n \beta_j^L \; , \quad \text{ for all } L=1,\ldots, n.
\ee
Complete them as $\alpha_{r+1}=\cdots =\alpha_n=0$. Then there exists a permutation $P\in \Pi_n$ so that 
\be
(\alpha_1,\ldots, \alpha_n)=(\beta_{P(1)},\ldots, \beta_{P(n)}).
\ee
\end{lemma}

This result is standard but we include the proof for completeness.

\begin{proof}
Let  $s_k$ denote the power sum polynomials in $n$ variables, $s_k=x_1^k+\cdots +x_n^k$, and 
\be
\tau_k=\sum_{i_1<i_2<\cdots<i_k}x_{i_1}x_{i_2}\cdots x_{i_k}.
\ee
We have the following Newton identities:
\bea
\tau_1&=&s_1, \\
\tau_k&=&(-1)^{k-1}\frac{1}{k}\left(s_k-\tau_1s_{k-1}+\cdots + (-1)^{k-1}\tau_{k-1}s_1\right), \nonumber
\eea
which allow us to write all $\tau_1,\ldots, \tau_n$ as  polynomials on $s_1,\ldots, s_n$.

Now we consider a new variable $X$ and 
\be
f(X)=(X-x_1)(X-x_2)\cdots (X-x_n).
\ee 
By Viet\'a's formula, we have
\be
f= X^n-\tau_1X^{n-1}+\tau_2X^{n-2}+\cdots +(-1)^n\tau_n.
\ee
We define now $f_\alpha$ as $f$ with $x_i=\alpha_i$ and similarly $f_\beta$. By the above and the hypothesis, both are polynomials in $X$ with the same coefficients (hence the same polynomial) but one has roots $\alpha_1,\ldots,\alpha_n$ (with the repetitions given by the multiplicities) and the other $\beta_1,\ldots,\beta_n$. The conclusion follows. 
\end{proof}

A trivial corollary of this lemma is the following. 

\begin{cor}\label{cor:0s}
Let $\alpha_1,\ldots \alpha_r\in \C\backslash\{0\}$. Then there exists $L\le r$ so that $\sum_{i=1}^r \alpha_i^L\not = 0$.
\end{cor}

Another corollary for the proportional case is the following.

\begin{cor}
\label{cor:prop}
For each $\alpha \in \{1,\ldots, \Gamma\}$, consider the vectors $(\mu_{\alpha,1},\ldots, \mu_{\alpha, r_{\alpha}})$ and $(\nu_{\alpha,1},\ldots, \nu_{\alpha, r_{\alpha}})$ (if they do not have the same length, we complete them with zeros), and let $L_0=\sum_\alpha r_\alpha$. Assume that  for each $L\le L_0^2$ there exists a constant $m_L\neq 0$ independent of $\alpha$ such that
\begin{equation}\label{eq:lem4}
\sum_{i=1}^{r_\alpha} \mu_{\alpha, i}^L=m_L \sum_{i=1}^{r_\alpha} \nu_{\alpha, i}^L \quad \forall \alpha .
\end{equation}
Then for all $L\in \N$ there exists a $m_L$ so that (\ref{eq:lem4}) holds.
\end{cor}

\begin{proof}
Since $m_L\not = 0$,  Eq.~\eqref{eq:lem4} is equivalent to 
\begin{equation}\label{eq:lem4-2}
\sum_{i=1}^{r_\alpha} \mu_{\alpha, i}^L\sum_{j=1}^{r_\beta} \nu_{\beta, j}^L=\sum_{i=1}^{r_\alpha} \nu_{\alpha, i}^L\sum_{j=1}^{r_\beta} \mu_{\beta, j}^L \quad \forall \alpha, \beta
\end{equation}
Since this holds for all $L\le L_0^2$ and $L_0^2\ge r_{\alpha}r_{\beta}$, Lemma \ref{lem1} ensures that the vectors $(\mu_{\alpha,i}\nu_{\beta,j})_{ij}$ and $(\nu_{\alpha,i}\mu_{\beta,j})_{ij}$ are related by a permutation, which in turn implies that Eq.~\eqref{eq:lem4-2} holds for all $L\in \N$, and thus Eq.~\eqref{eq:lem4} also holds for all $L\in \N$.
\end{proof}

We are now ready to prove \cref{thm:mps}. 

\begin{proof}[Proof of \cref{thm:mps}]
Consider $\{A_i\in \mathcal{M}_{D}\}$, $\{C_i\in \mathcal{M}_{D'}\}$  that satisfy \eqref{eq:hyp}. 
First we block sites into groups of size $D!D'!$ so that we can assume  without loss of generality  that both $|\psi_A\ra$ and  $|\psi_C\ra$ are in canonical form with each block being injective. 
From now on we will work with the grouped sites so the real system size $L'$ will be of the form $L'=D!D'! L$ for some $L$. In the following we will consider just such $L$'s.

To each vector $|\psi_A^L\ra$ and $|\psi_C^L\ra $ we  apply $e^{\mi\phi}U \cdot U$ if needed to the different blocks so that those related by such relation become the same (and the phase $e^{\mi\phi}$ is absorbed in the $\lambda$'s, which now become complex numbers). Thus, by applying  $VA_iV^\dagger$ and $WC_i W^\dagger$ with $V, W$ unitaries we can assume without loss of generality that the $A$'s and $C$'s take the form
\bea
A_i= \oplus_{j\in J} \left(A^j_i\otimes {\rm diag}(\lambda^j_1,\ldots,\lambda^j_{r_j})\right) \\
C_i= \oplus_{k\in K}\left( C^k_i\otimes {\rm diag}(\mu^k_1,\ldots,\mu^k_{s_k})\right) , 
\eea
with $\lambda$'s and $\mu$'s complex numbers, and neither of the $A^j_i$'s are related to each other by $e^{\mi\phi}U\cdot U^\dagger$, and the same holds for the $C$'s.

Define the set $K'\subset K$ as follows: $k\in K'$ if there exist $j\in J$, a phase $\phi$ and a unitary $U$ so that $e^{\mi\phi}UA^j_iU^\dagger = C^k_i$ for all $i$. We want to see that $K'=K$.
If this were not the case, by Theorem \ref{thm:li} applied to $\{A^j_i\}_i$ with $j\in J$ and $\{C^k_i\}_i$ with $k\in K\backslash K'$, we would obtain that for all $L\ge 3(D+D')^5$ the set of vectors $\ket{\psi_{A^j}^L}$, $\ket{\psi_{C^k}^L}$ with $j\in J, k\in K\backslash K'$ is linearly independent. On the other hand, by the hypothesis (Eq.~\eqref{eq:hyp}), 
\be
\sum_j\left(\sum_l\left(\lambda^j_l\right)^L\right)\ket{\psi_{A^j}^L}= m_L\sum_k\left(\sum_m\left(\mu^k_m\right)^L\right)\ket{\psi_{C^k}^L} .
\ee
Thus, to obtain a contradiction, we only need to show that for some $k\not \in K'$ there exists an $L$ with 
$3(D+D')^5\le L\le3(D+D')^6$ 
so that $\sum_m (\mu^k_m)^L\neq 0$. This is true by \cref{cor:0s} applied to $\alpha_m=(\mu^k_m)^{3(D+D')^5} $. So $K=K'$. 

Interchanging the roles of $A$ and $C$ one obtains that $K=J$, and,  again by applying  $VA_iV^\dagger$ and $WC_i W^\dagger$ with $V, W$ unitaries,  we obtain the form
\bea
&&A_i= \oplus_{j\in J} \left(A^j_i\otimes {\rm diag}(\lambda^j_1,\ldots,\lambda^j_{r_j})\right) \\
&&C_i= \oplus_{j\in J}\left( A^j_i\otimes {\rm diag}(\mu^j_1,\ldots,\mu^j_{s_j})\right) .
\eea
Now from the linear independence of the vectors $\ket{\psi_{A^j}^L}$ (\cref{thm:li}), and 
 the hypothesis (Eq.~\eqref{eq:hyp}), it follows that for all $j$,
\be
\sum_l\left(\lambda^j_l\right)^L=m_L\sum_m\left(\mu^j_m\right)^L
\ee 
for all $3(D+D')^5\le L\le3(D+D')^6$. Applying \cref{cor:prop} to $\alpha_l= (\lambda^l_m)^{3(D+D')^5}$ and $\beta_m= (\mu^k_m)^{3(D+D')^5}$ we obtain that 
$\ket{\psi_A^L}\propto\ket{\psi_C^L}$ for all $L$'s which are multiples of $3(D+D')^5$. Recalling that the total system size is $L'=D!D'! L$, we obtain the result.
\end{proof}

\section{Proof of the Main result (Theorem 1)}
\label{sec:proof}

We are finally in a position to prove the main result of this paper. 

\begin{proof}[Proof of \cref{thm:main}]
Assume that \cref{thm:main} is not true, that is, for every tensor  $A$ such that $\rho_A^L\geq 0$ for all $L$,  there exists a tensor $B$ of finite size so that $\rho_A^L\propto \sigma_B^L$ for all $L$. 
We will use this  fact to design an algorithm that solves an undecidable problem.

So take a family of matrices $\{A_i\}$ which are a valid input to \cref{prob:undec} (that is, either Case 1 or Case 2 holds). Then consider the following algorithm:

\

\noindent $D'=1$, $L=1$\\

\noindent \texttt{while} not halt \texttt{do} \\

\indent Decide if there exists $|\psi_C^L\ra$ with $C_i$ of form \eqref{eq:pure-purification} and size $D'$ that verifies Eq.~\eqref{eq:hyp}.\\

\indent \texttt{if} it exists, output `Case 1' and halt.\\

\indent \texttt{else} diagonalize  $\rho_A^L$.\\

\indent \indent \texttt{if}  ${\rm mineig}(\rho_A^L)<0$ output  `Case 2' and halt. \\

\indent \indent \texttt{else}  $D'=D'+1$, $L=L+1$\\

\noindent \texttt{end while}

\

Now we show that this algorithm always halts and correctly decides whether $\{A_i\}$ verifies Case 1 or 2. 
If $\{A_i\}$ is in Case 2,  there exists a finite $L_0$ so that  $\rho_A^L\not \ge 0$ for all $L\ge L_0$. Therefore, after $L_0$ iterations the algorithm will find it and halt.

If $\{A_i\}$ is in Case 1,  since we are assuming that Theorem \ref{thm:main} is not true, there exists an MPDO in purification form with matrices  $\{B_{i,e}\in \mathcal{M}_{\sqrt{D'}}\}$ for some finite $D'$ so that \eqref{eq:main-thm} holds. 
By virtue of \cref{thm:mps}, deciding whether such $\{B_{i,e}\}$ exist is a decidable problem, since it consists of deciding whether a system of finitely many polynomial equations (whose unknowns are the entries of $B$)  has a solution over the complex numbers. 
Thus, the algorithm will check it and halt after $D' $ iterations.

%
%
%
%

Finding an algorithm to solve Problem \ref{prob:undec} shows that it is a decidable problem, which contradicts Lemma \ref{lem:undec} and finishes, by contradiction, the proof of Theorem \ref{thm:main}. 
\end{proof}

\section{Conclusions and Outlook}
\label{sec:conclusions}

In conclusion, we have shown that there exist mixed states with a TI MPDO description valid for all system sizes, but for which there does not  exist  a purification which is also translational invariant and valid for all system sizes. 
To prove this result we have relied on the notion of undecidability and on the uniqueness of canonical forms of MPS. 
It is an interesting open question for further investigation whether with this result or with similar techniques, one can attack the old question about the existence of Finitely Correlated States which are not  C$^*$ Finitely Correlated \cite{Fa92}.
\section*{Acknowledgements}

GDLC and JIC acknowledge support from SIQS.
TSC is supported by the Royal Society.
MMW acknowledges support from the CHIST- ERA/BMBF project CQC.
 The opinions expressed in this publication are those of the authors and do not necessarily reflect the views of the John Templeton Foundation.
DPG acknowledges support from MINECO (grant MTM2014-54240-P and ICMAT Severo Ochoa project SEV-2015-0554) and Comunidad de Madrid (grant QUITEMAD+-CM, ref. S2013/ICE-2801).
This work was made possible through the support of grant $\#$48322 from the John Templeton Foundation. 
This project has received funding from the European Research Council (ERC) under the European Union's Horizon 2020 research and innovation programme (grant agreement No 648913).



\begin{thebibliography}{21}
\bibitem{fn1}
$A\geq0$ denotes that $A$ is hermitian and positive semi-definite.
\end{thebibliography}

\end{document}